\documentclass[11pt,letterpaper,onecolumn]{article}

\usepackage[T1]{fontenc}
\usepackage[utf8]{inputenc}
\usepackage{lmodern}
\usepackage{microtype}
\usepackage[margin=1in]{geometry}
\usepackage{setspace}
\usepackage{authblk}
\usepackage{mathtools}
\usepackage{amssymb}
\usepackage{amsthm}
\usepackage{aliascnt}
\usepackage{booktabs}
\usepackage{graphicx}

\usepackage{xcolor}
\usepackage[
  backend=biber,
  style=numeric-comp,
  sorting=nyt,
  sortcites=true,
  giveninits=true,
  maxbibnames=99,
  doi=true,
  url=true,
  eprint=true,
  isbn=false
]{biblatex}
\usepackage[hidelinks]{hyperref}
\usepackage[capitalise,nameinlink,noabbrev]{cleveref}

\theoremstyle{plain}
\newtheorem{theorem}{Theorem}[section]
\newaliascnt{lemma}{theorem}
\newtheorem{lemma}[lemma]{Lemma}
\aliascntresetthe{lemma}
\newaliascnt{proposition}{theorem}
\newtheorem{proposition}[proposition]{Proposition}
\aliascntresetthe{proposition}
\newaliascnt{corollary}{theorem}
\newtheorem{corollary}[corollary]{Corollary}
\aliascntresetthe{corollary}
\newaliascnt{claim}{theorem}

\aliascntresetthe{claim}
\newaliascnt{conjecture}{theorem}
\newtheorem{conjecture}[conjecture]{Conjecture}
\aliascntresetthe{conjecture}

\theoremstyle{definition}
\newaliascnt{definition}{theorem}

\aliascntresetthe{definition}

\theoremstyle{remark}
\newaliascnt{remark}{theorem}
\newtheorem{remark}[remark]{Remark}
\aliascntresetthe{remark}
\newaliascnt{observation}{theorem}

\aliascntresetthe{observation}

\usepackage{thmtools}

\crefname{appendix}{appendix}{appendices}
\Crefname{appendix}{Appendix}{Appendices}

\newif\ifshowcomments
\showcommentsfalse

\definecolor{AineshComment}{HTML}{F05A00}
\definecolor{AlexComment}{HTML}{8A2BE2}
\definecolor{HannaComment}{HTML}{E0008A}
\definecolor{WilliamComment}{HTML}{0077E6}
\definecolor{AlekComment}{HTML}{00A651}

\newcommand{\defn}[1]{\textbf{\textit{#1}}}
\newcommand{\imod}[2]{\left[#1\right]_{#2}}  
\newcommand{\hreal}{h^{\mathrm{real}}}
\newcommand{\hmod}{h^{\mathrm{mod}}}
\newcommand{\maxload}{\operatorname{maxload}}

\usepackage{todonotes}
\hypersetup{
  pdftitle={Lower Bounds for Linear Hashing via Arithmetic Kakeya},
  pdfauthor={Ainesh Bakshi, Alex Conway, Hanna Komlós, William Kuszmaul, Alek Westover},
  pdfdisplaydoctitle=true,
  bookmarksnumbered=true
}

\title{Lower Bounds for Linear Hashing via Arithmetic Kakeya}

\author[1]{Ainesh Bakshi}
\author[2]{Alex Conway}
\author[3]{Hanna Komlós}
\author[4]{William Kuszmaul}
\author[5]{Alek Westover}

\affil[1]{New York University (\href{mailto:ainesh@nyu.edu}{\texttt{ainesh@nyu.edu}})}
\affil[2]{Cornell University (\href{mailto:me@ajhconway.com}{\texttt{me@ajhconway.com}})}
\affil[3]{Max Planck Institute for Informatics (\href{mailto:hkomlos@gmail.com}{\texttt{hkomlos@gmail.com}})}
\affil[4]{Carnegie Mellon University (\href{mailto:kuszmaul@cmu.edu}{\texttt{kuszmaul@cmu.edu}})}
\affil[5]{Redwood Research (\href{mailto:alek.westover@gmail.com}{\texttt{alek.westover@gmail.com}})}
\date{}

\begin{document}

\maketitle

\begin{abstract}
Affine modular linear hashing is one of the simplest classical hash families.
For a prime \(p>u\), the hash function is obtained by choosing \(s,t\) uniformly from \(\mathbb Z_p\) and mapping each key \(x\in\{0,\ldots,u-1\}\) to one of \(n\) bins by
\[
\hmod_{s,t}(x)
\coloneqq
\bigl[(sx+t)\bmod p\bigr]\bmod n.
\]
Despite its simplicity, the maximum load of linear hashing remains poorly understood.
For \(n\) keys hashed into \(n\) bins, the best known upper bound is \(O((n\log n)^{1/3})\), whereas the best known lower bound is only \(\Omega(\log n/\log\log n)\).
We prove a lower bound of \(\exp(\Omega(\log n/\log\log n))\) for universes of size \(n^{1+o(1)}\).
Surprisingly, there is a key set for which this load holds not just in expectation, but for every random seed.
The proof is driven by two simple reductions: one transfers lower bounds from a real version of linear hashing to modular linear hashing, and the other transfers arithmetic Kakeya constructions to real hashing.
We further show that, for sufficiently large \(p\), the expected maximum loads in the modular and real settings are essentially the same, giving an alternative route to an \(n^{1/3+o(1)}\) upper bound.
Finally, we show that any uniform subpolynomial upper bound for either setting would imply a polynomial-length arithmetic Kakeya conjecture and hence the Kakeya conjecture for upper Minkowski dimension.
\end{abstract}

\begingroup
\renewcommand\thefootnote{}
\footnotetext{Independent concurrent work of Zamir~\cite{zamir2026linearhashingawesome} obtains the same super-polylogarithmic lower bound for linear hashing. The two works differ in their surrounding frameworks and additional results. }
\endgroup

\section{Introduction}
\label{sec:intro}

Hash functions are used throughout computer science for load distribution.
Among other applications, they are used by hash tables to distribute keys across locations~\cite{KnuthVol3,dietzfelbinger-et-al1994dynamic,pagh-rodler2004cuckoo}, by low-associativity caches to distribute cache blocks among sets~\cite{seznec1993skewed,kharbutli-solihin-lee2005cache,sanchez-kozyrakis2010zcache}, and by randomized and distributed load balancers to assign jobs or clients to servers~\cite{azar-et-al1999balanced,mitzenmacher2001power,mirrokni-thorup-zadimoghaddam2018consistent}.
In these applications, it is important that hash functions distribute keys evenly among bins, and one of the most important measures of this is the maximum load.
For example, in a hash table with separate chaining, a lookup or update may scan every key in one chain, and the maximum load therefore controls the worst-case lookup time.

An ideal fully random hash function has maximum load \(\Theta(\log n/\log\log n)\) with high probability when $n$ items are hashed to $n$ locations~\cite{raab-steger1998balls}.
However, fully random hash functions are not practical because they require high space and time to evaluate.
As a result, there is a long line of research both in theory and practice which looks for hash families that give comparable guarantees to fully random hash families, but which are small and fast to evaluate~\cite{carter1977universal,dietzfelbinger1996universal,Siegel2004ExtremelyRandomConstantTime,PaghPagh,tabulationoriginal,thorup2017fast}.

This paper studies one such hash family, known as \defn{affine modular linear hashing} (or linear hashing, for short, when it is unambiguous).
A modular linear hash function maps keys from a universe \(U = [u]\) to \([n]\)\footnote{Throughout the paper we write \([k]\coloneqq \{0,1,\dots,k-1\}\).} using modular arithmetic with a fixed prime \(p > u\). The hash function is parameterized by random $s, t \in \mathbb{Z}_p$, and maps each key $x$ to
\[
\hmod_{s,t}(x) = \imod{\imod{s x+t}{p}}{n},
\]
where $\imod{a}{b}$ is the residue of $a$ modulo $b$. Note that this hash function requires only $O(\log p)$ bits to store, and can be evaluated in $O(1)$ time.

Linear hashing was first studied in a classic 1979 paper by Carter and Wegman \cite{carter1977universal}, who showed that the underlying affine maps over \(\mathbb Z_p\) form a \(2\)-universal and pairwise-independent family.\footnote{After the final reduction to \([n]\), the marginals are only nearly uniform, but two distinct keys still collide with probability \(O(1/n)\), which is all that we use here.}
This pair-collision bound gives a generic \(O(\sqrt n)\) expected maximum-load bound when a set \(X\) of \(n\) keys is hashed into \(n\) bins.
Indeed, the expected number of colliding pairs of keys is \(O(n)\), while a bin of load \(L\) contributes \(\binom{L}{2}\) colliding pairs.

Subsequent work showed that the \(O(\sqrt n)\) bound is tight for general pairwise-independent hash families \cite{alon1997linear}.
However, it remained possible that linear hashing specifically achieves a better bound.
Nearly four decades after Carter and Wegman's original analysis \cite{carter1977universal}, Knudsen showed this to be the case: in a paper titled \emph{Linear Hashing is Awesome}, he proved that the expected maximum load is \(O((n\log n)^{1/3})\) \cite{knudsen2016linear,knudsen2019linear}.

Perhaps surprisingly, on the lower bound side, essentially nothing is known. It has remained possible that, in terms of expected maximum load, linear hashing might even match the $O(\log n / \log \log n)$ bound achieved by fully random hash functions. If so, it would be the first known explicit hash function with $O(\log u)$ description bits (where $u$ is the universe size) to achieve this bound \cite{celis2013balls,meka2014fast,kuszmaul2022hash}.

\paragraph{This paper: A new lower bound via two reductions. }
The main result of this paper is a lower bound showing that the expected maximum load using modular linear hashing is
\begin{equation}
    \exp(\Omega(\log n / \log \log n)).
\label{eq:singleloglog}
\end{equation}
This represents an \emph{exponential} improvement over the previous best lower bound of \(\Omega(\log n / \log \log n)\) \cite{knudsen2016linear}, which, when \(u \ge n^2\), follows simply by considering a random key set \(X\)~\cite[Theorem~5.1]{celis2013balls}.
Interestingly, our proof of \eqref{eq:singleloglog} actually establishes a stronger fact: there exists a set \(X\) of \(n\) elements such that, \emph{for every pair \(s, t \in \mathbb{Z}_p\)}, the maximum load of \(\hmod_{s, t}\) on \(X\) is \(\exp(\Omega(\log n / \log \log n))\).
Thus the lower bound holds not just in expectation, but even pointwise.

The new lower bound is achieved by chaining together two simple reductions, each of which allows us to transfer known bounds from one problem to another.
Therefore, we view another contribution of the paper to be establishing the connections between these seemingly unrelated problems.

The first reduction relates the maximum load under \emph{modular} linear hashing (Problem 1) to the maximum load of a continuous version of the problem, \defn{real linear hashing} (Problem 2).
We define the hash functions $\hreal_a(x)$ parameterized by a random real number $a \in [0, 1)$ as
\[\hreal_a(x) = \lfloor n \{ax\} \rfloor,\]
where $\{ax\}$ denotes the fractional part of $ax$.\footnote{Real linear hashing was introduced in an earlier preprint by one of the authors~\cite{westover2023relationship}. \Cref{thm:problem2-to-problem1,thm:stronger-connection} can be viewed as strengthening results in that work.}
We show that any lower bound in the real-valued setting that holds for \emph{all} $a \in [0, 1)$ implies the same lower bound in the modular-valued setting that holds for \emph{all} $s,t \in \mathbb{Z}_p$.
This reduction allows us to reinterpret known results in combinatorial number theory \cite{konyagin-ruzsa-schlag2000} to immediately obtain a lower bound of
\[\exp(\Omega(\log n / (\log \log n)^2)).\]

To improve this bound to \eqref{eq:singleloglog}, we relate the real-valued version of linear hashing to a third problem, known as the \defn{arithmetic Kakeya problem} (Problem 3).
Problem 3, which is also studied in combinatorial number theory, seeks, for given integers \(K,D\), to construct an integer set \(X\) that is as small as possible while containing \(K\)-term arithmetic progressions with every positive difference \(d\in[D]\).
We show that known constructions for Problem 3 imply a new $\exp(\Omega(\log n / \log \log n))$ lower bound for Problem 2.
Chaining together both of our reductions, we obtain the same lower bound for Problem 1.
Dhar and Dvir~\cite{dhar-dvir2024linear} previously connected linear hashing and Kakeya sets in a different but related problem, as we discuss in \Cref{sec:related}.

\paragraph{Further strengthening the relationship between Problems 1 and 2. }We also extend our reduction between Problems 1 and 2 to obtain an \emph{equivalence result}: We show, for $p$ sufficiently large, the expected maximum load under modular linear hashing is \emph{the same} as the expected maximum load under real linear hashing differ by at most $1$.

As an immediate consequence, we are able to obtain a new path to Knudsen's $n^{1/3 + o(1)}$ upper bound (albeit with a worse $o(1)$ term).
Specifically, we can take a known upper bound for (the average-case version of) Problem 2 \cite{konyagin-ruzsa-schlag2000} and transform it directly into an upper bound for Problem 1.

\paragraph{A conditional obstacle to subpolynomial upper bounds.}
Finally, our reductions also present a natural obstacle to obtaining significantly better \emph{upper bounds} for either of Problems 1 or 2 and, specifically, for the expected maximum load of modular linear hashing.
We show that a subpolynomial bound would imply a polynomial-length version of the arithmetic Kakeya conjecture, which is known to imply  the \defn{Kakeya conjecture for upper Minkowski dimension}.
A \defn{Besicovitch set} in \(\mathbb R^d\) is a compact measure-zero set containing a unit line segment in every direction, and this conjecture asserts that every such set has upper Minkowski dimension \(d\) \cite{bourgain1991remarks,bourgain1993distribution,green-ruzsa2019arithmetic}; see \Cref{sec:hashing-implies-kakeya}.

\paragraph{Roadmap.}
We discuss related work in the remainder of this section.
In \Cref{sec:problems}, we formally introduce the three problems described above and present the reductions.
In \Cref{sec:uppers}, we strengthen the relationship between Problems 1 and 2 to include upper bounds.
In \Cref{sec:hashing-implies-kakeya}, we show that subpolynomial bounds for either of Problems 1 or 2 imply the Kakeya conjecture for upper Minkowski dimension.

\subsection{Related Work}
\label{sec:related}

\paragraph{Affine modular linear hashing.}
Carter and Wegman introduced affine modular linear hashing as a particularly simple universal family~\cite{carter1977universal}.
The pair-collision bound inherited from the underlying field-valued family gives the generic \(O(\sqrt n)\) expected maximum-load bound, which Knudsen improved to \(O((n\log n)^{1/3})\)~\cite{knudsen2016linear,knudsen2019linear}.
Knudsen's analysis gives the same bound for multiply--shift hashing, a closely related constant-time family over power-of-two word universes introduced by Dietzfelbinger et al.~\cite{dietzfelbinger-et-al1997multiply-shift}.
Under a distributional rather than worst-case input model, Chung, Mitzenmacher, and Vadhan showed that every \(2\)-universal family, including standard Carter--Wegman families, behaves like fully random hashing in several applications, including chained hashing, when each key contributes enough conditional R\'enyi entropy~\cite{chung-mitzenmacher-vadhan2013entropy}.
Westover's preliminary version of this paper also studied shiftless, composite-modulus, and real variants of affine modular linear hashing~\cite{westover2023relationship}.

\paragraph{Finite-field linear hashing.}
A separate line of work also confusingly referred to as ``linear hashing'' studies hashing by a uniformly random linear map \(h\colon \mathbb F_2^d\to\mathbb F_2^r\), rather than by the scalar affine map over \(\mathbb Z_p\) studied here.
For \(u=2^d\) and \(n=2^r\), a dense matrix representation uses \(\Theta(\log u\log n)\) bits and straightforward evaluation uses \(O(\log u)\) word-level XOR operations; when \(u\) and \(n\) are polynomially related, these become \(\Theta((\log u)^2)\) bits and \(O(\log u)\) operations~\cite{jaber-kumar-zuckerman2025linear}.
For \(n\) keys and \(n\) bins, the expected maximum-load bound progressed from \(O(n^{1/4})\) to \(2^{O(\sqrt{\log n})}\), \(O(\log n\log\log n)\), and \(O(\log n)\)~\cite{markowsky-carter-wegman1978analysis,mehlhorn-vishkin1984simulations,alon1997linear,babka2018largest-bins}.
Jaber, Kumar, and Zuckerman recently obtained the optimal \(O(\log n/\log\log n)\) bound~\cite{jaber-kumar-zuckerman2025linear}.

\paragraph{Connections between linear hashing and Kakeya.}
Recently, Dhar and Dvir proved strong two-sided \(\ell_\infty\)-load guarantees for finite-field linear hashing using finite-field Kakeya bounds~\cite{dhar-dvir2024linear}.
They observed that extending these guarantees to affine modular linear hashing would imply a strong generalization of the full arithmetic Kakeya problem, since its buckets are modular arithmetic progressions in every difference.
In contrast, we show that much weaker subpolynomial single-seed maximum-load upper bounds imply the weaker polynomial-length arithmetic Kakeya conjecture, which nevertheless suffices to establish the Kakeya conjecture for upper Minkowski dimension.

\paragraph{Efficient hashing with near-random guarantees.}
Wegman and Carter extended affine modular linear hashing to a degree-\((k-1)\) polynomial construction, which gives \(k\)-wise independent hashing over finite fields using \(O(k\log u)\) description bits and \(O(k)\) field operations per evaluation~\cite{wegman_carter_1981_new_hash_functions}.
Taking \(k=\Theta(\log n/\log\log n)\), a standard limited-independence argument gives maximum load \(O(\log n/\log\log n)\) with high probability.
For polynomial-size universes, Celis et al. reduced the description length for the same guarantee to \(O(\log n\log\log n)\) bits, and Meka et al. retained this description length while reducing evaluation to \(O((\log\log n)^2)\) word operations~\cite{celis2013balls,meka2014fast}.
Another line of work trades polynomial space for constant evaluation time: Siegel constructed highly independent families and established the governing time--space lower bounds, while Christiani, Pagh, and Thorup gave word-RAM constructions close to this tradeoff~\cite{Siegel2004ExtremelyRandomConstantTime,christiani2015independence}.
Pagh and Pagh instead obtained constant-time, near-entropy-optimal hashing that is fully random with high probability on any fixed \(n\)-key set~\cite{PaghPagh}.
Simple tabulation hashing, despite being only \(3\)-wise independent, achieves the optimal maximum load with constant evaluation time, and newer tabulation schemes provide local-uniformity guarantees and efficient full randomness on fixed sets~\cite{tabulationoriginal,bercea2023locally}.

\paragraph{Real linear hashing / uniform dilation.}
Montgomery records the underlying uniform-dilation question and attributes it to Koml\'os and Ruzsa \cite[Problem 17(4)]{montgomery1994ten}.
This led to the work of Konyagin, Ruzsa, and Schlag discussed above~\cite{konyagin-ruzsa-schlag2000}, which builds on the ideas of Alon and Peres~\cite{alon-peres1992uniform}.
They also posed a modular analogue concerning concentration in residue classes, proved an \(O(\sqrt n)\) upper bound, and conjectured an \(O(\log n)\) bound \cite[Question 7.1]{konyagin-ruzsa-schlag2000}.

\paragraph{Arithmetic Kakeya.}
The arithmetic Kakeya problem is a discrete analogue of the geometric one: arithmetic progressions with every common difference play the role of line segments in every direction.
Bounds for such progression-rich sets feed into arithmetic-projection arguments that yield dimension bounds for Besicovitch sets.
This arithmetic-projection approach to the Kakeya problem originated in work of Bourgain \cite{bourgain1991remarks,bourgain1993distribution} and was developed through the sums--differences estimates of Katz and Tao \cite{katz-tao1999arithmetic-projections,katz2002newbounds}.
Ruzsa later posed the progression formulation defining \(F_k(D)\) as the minimum size of a finite set of integers containing a \(k\)-term arithmetic progression with every positive common difference \(d\in[D]\), without identifying its Kakeya connection \cite[Conjecture 4.2]{ruzsa2005sumsets}.
Green and Ruzsa proved that this formulation is equivalent to the Katz--Tao projection conjecture, an entropy inequality, finite-field progression conjectures, and a prime-multiples problem of Erd\H{o}s and Selfridge \cite[Theorem 1.1]{green-ruzsa2019arithmetic}.
The conjecture remains open.
They also constructed sets showing that
\[
 \lim_{D\to\infty}\frac{\log F_k(D)}{\log D}
 \leq
 1-\frac{c}{\log\log k}
\]
for an absolute constant \(c>0\) \cite[Theorem 1.2]{green-ruzsa2019arithmetic}.
Already for \(k=3\), the arithmetic-projection bound of Katz and Tao gives \(F_3(D)\geq D^{6/11-o(1)}\), while the best general projection exponent remains their value \(\alpha=1.6751308\ldots\), the largest root of \(\alpha^3-4\alpha+2=0\) \cite{katz-tao1999arithmetic-projections,katz2002newbounds,tao2025sum-difference}.
Bourgain showed that it already suffices for the Kakeya conjecture for upper Minkowski dimension to prove the weaker polynomial-length statement with \(k=D^{\eta}\) for every fixed \(\eta>0\) \cite{bourgain1991remarks,bourgain1993distribution,green-ruzsa2019arithmetic}.
Cowen-Breen, Karangozishvili, Varadarajan, and Wang established further equivalent pattern formulations and strengthened the geometric consequence of the full arithmetic conjecture from Minkowski to packing dimension \cite{cowen-breen-et-al2020pattern}.

\section{Connecting Together Three Problems}\label{sec:problems}

This section defines the three problems that we study in this paper, discusses the known bounds for each of them, and then presents two simple reductions that allow us to obtain new lower bounds for Problems 1 and 2.

\paragraph{Problem 1: Expected maximum load for modular linear hashing.}
Fix parameters \(0 < n \le u  < p\), where \(p\) is prime, and let \(U=[u]\) be the key universe.

For an integer \(z\), write \(\imod{z}{m}\in \{0,\ldots,m-1\}\) for its remainder modulo \(m\).
Then the affine modular linear-hashing family is the collection of hash functions \(\hmod_{s,t}:U \rightarrow [n]\) defined by
\[
 \hmod_{s,t}(x)\coloneqq
 \imod{\imod{s x+t}{p}}{n},
 \qquad s,t\in[p].
\]
The additive parameter \(t\) changes the maximum load by at most a factor of two\footnote{
Indeed, fix \(s,t\), and write \(r_x\coloneqq\imod{sx}{p}\) and \(r'_x\coloneqq\imod{r_x+t}{p}\).
For every key \(x\), either \(r'_x=r_x+t\) or \(r'_x=r_x+t-p\).
Consequently, the keys in one bin \(j\) under \(\hmod_{s,t}\) lie in at most the two bins \(\imod{j-t}{n}\) and \(\imod{j-t+p}{n}\) under \(\hmod_{s,0}\), while the keys in one bin \(j\) under \(\hmod_{s,0}\) lie in at most the two bins \(\imod{j+t}{n}\) and \(\imod{j+t-p}{n}\) under \(\hmod_{s,t}\).
Thus, for every key set \(X\),
 $\frac{1}{2}\maxload_X(\hmod_{s,0})
 \leq
 \maxload_X(\hmod_{s,t})
 \leq
 2\maxload_X(\hmod_{s,0})$.
}.
Because the bounds in this paper are insensitive to constant factors, we may therefore set \(t=0\), and write
\[
 \hmod_s\coloneqq\hmod_{s,0}.
\]
As a convention we will use \(\mathcal{H}_p^{n}\) to denote the uniform distribution over functions \(\hmod_s\) (i.e., with \(s\) uniformly random in \([p]\)).

For any set \(X\subseteq U\) of keys, and for any hash function \(h:U\to[n]\), define
\[
 \maxload_X(h)\coloneqq
 \max_{j\in[n]}\big|\{x\in X:h(x)=j\}\big|.
\]
The first problem that we consider, which we call the \defn{Modular Linear Hashing Problem (Problem 1)}, is to determine worst-case bounds for the expected maximum load of $\hmod_s$ on any set $X \subseteq U$ of $n$ keys, where the expectation is over the choice of function $\hmod_s$ from $\mathcal{H}_p^n$.
That is, we wish to bound
\begin{equation}
\max_{X \subseteq U, |X| = n} \mathbb E_{\hmod_s \sim \mathcal{H}_p^{n}}[\maxload_X(\hmod_s)].
\label{eq:modular-linear-hashing-problem}
\end{equation}

On the upper-bound side, Knudsen \cite{knudsen2016linear,knudsen2019linear} shows that \eqref{eq:modular-linear-hashing-problem} is at most $O((n\log n)^{1/3})$.
\begin{theorem}[\cite{knudsen2016linear,knudsen2019linear}]
Let \(0<n\leq u<p\), where \(p\) is prime, and let \(X\subseteq U\) be a set of \(n\) keys.
Then, for \(\hmod_s\sim\mathcal{H}_p^n\), the expected maximum load of \(\hmod_s\) on \(X\) is at most \(O((n\log n)^{1/3})\).
\label{thm:knudsen-upper-bound}
\end{theorem}

In this paper, we will be primarily interested in lower bounds for \eqref{eq:modular-linear-hashing-problem}.
Here, the only known non-trivial bound is the \(\Omega(\log n / \log \log n)\) bound that applies to \emph{any} family of hash functions when \(u \ge n^2\) (this lower bound can be seen by taking \(X\) to be uniformly random)~\cite[Theorem~5.1]{celis2013balls}.

\paragraph{Problem 2: Minimum max load for real linear hashing.}
To define the second problem, we again consider parameters $0 < n \le u$, and let $U=[u]$ be the key universe.
Write \(\mathbb T=\mathbb R/\mathbb Z\) identified with \([0,1)\), and let \(\{y\}\in[0,1)\) denote the fractional part of a real number \(y\).

The real linear-hashing family is the (infinite) collection of hash functions $\hreal_a:U \to [n]$ defined by
\[
 \hreal_a(x)\coloneqq\left\lfloor n\{ax\}\right\rfloor,
 \qquad a\in\mathbb T.
\]
As a convention we will use $\mathcal{H}_{\mathbb T}^n$ to denote the uniform distribution over $\hreal_a$ (i.e., with $a$ uniformly random in $\mathbb T$).

The second problem that we consider, which we call the \defn{Best-Case Real Linear Hashing Problem (Problem 2)}, is to determine the max load 
\begin{equation}
K(n, u) \coloneqq \max_{X \subseteq U, |X| = n} \min_{a \in \mathbb T} \maxload_X(\hreal_a).
\label{eq:krs-quantity}
\end{equation}
Problem 2 has been studied in combinatorial number theory as a question about whether all integer sets \(X\) have some dilate \(aX\) with reasonably uniform fractional parts~\cite{konyagin-ruzsa-schlag2000,montgomery1994ten,green2025open}.
Note that, here, unlike in Problem 1, we maximize over key sets \(X\) the \emph{best-case} max load across all seeds \(a\in\mathbb T\), rather than the average-case max load for a uniformly random seed \(a\in\mathbb T\).
It turns out, however, that this distinction is mostly immaterial to the known bounds: the best known upper and lower bounds for Problem 2 are the same in either case.

In the language defined above, the state of the art for Problem 2 is the following theorem of Konyagin, Ruzsa, and Schlag~\cite{konyagin-ruzsa-schlag2000}.
The precise upper bound we state comes from a combination with earlier work of Alon--Peres and Alon--Ruzsa \cite{alon-peres1992uniform,alon-ruzsa1999nonaveraging}.
\begin{theorem}[\cite{konyagin-ruzsa-schlag2000,alon-peres1992uniform,alon-ruzsa1999nonaveraging}]
For every sufficiently large \(n\), there exists a universe size \(u=u(n)=n^{1+o(1)}\) such that \(K(n,u)\geq\exp(\Omega(\log n/(\log\log n)^2))\).
Moreover, for every \(u\geq n\), \(K(n,u)\leq n^{1/3}\exp(O(\sqrt{\log n\log\log n}))\).
\label{thm:krs-bounds}
\end{theorem}

We note that this connection also allows us to transfer Knudsen's slightly better upper bound of \(O(n\log{n})^{1/3}\)~\cite{knudsen2016linear,knudsen2019linear} for modular linear hashing to the real hashing problem, which slightly improves the best known upper bound above.

\paragraph{Problem 3: A set with arithmetic progressions in every step size.}
The final problem that we study is the \defn{Arithmetic Kakeya Set Problem (Problem 3)}, which is the problem of constructing a small set of integers that contains long arithmetic progressions with every positive difference in \([D]\).

Katz and Tao conjectured a lower bound on the size of arithmetic Kakeya sets, known as the arithmetic Kakeya conjecture~\cite{katz2001newbounds,katz2002newbounds}.
\begin{conjecture}
\label{conj:arithmetic-kakeya}
For positive integers \(k\) and \(D\), let \(F_k(D)\) be the minimum size of a finite set of integers containing a \(k\)-term arithmetic progression with every positive common difference \(d\in[D]\).
Then
\begin{equation}
\label{eq:arithmetic-kakeya}
 \lim_{k\to\infty}
 \lim_{D\to\infty}
 \frac{\log F_k(D)}{\log D}
 =
 1.
\end{equation}
\end{conjecture}

Green and Ruzsa showed that there exist Kakeya sets for which the ratio in \Cref{eq:arithmetic-kakeya} converges to 1 very slowly~\cite{green-ruzsa2019arithmetic}, which we will use for the lower bound.
In this paper, we will need a slightly more explicit version of their results:
\begin{restatable}{theorem}{greenruzsatheorem}
For every sufficiently large \(n\), there are integers \(K=\exp(\Theta(\log n/\log\log n))\) and \(u=n^{1+o(1)}\), and an integer set \(X\subseteq[u]\) of size \(n\), such that \(X\) contains a \(K\)-term arithmetic progression with every positive common difference \(d\in[nK]\).
\label{thm:green-ruzsa}
\end{restatable}

For completeness, we prove \cref{thm:green-ruzsa} using the main technical lemma from \cite{green-ruzsa2019arithmetic} in \Cref{app:arithmetic-kakeya}.

\paragraph{Reduction 1: From real linear hashing to modular linear hashing.}
The main technical contribution of this paper consists of two simple reductions relating the three problems described above. We begin with a surprisingly simple reduction which translates any lower bound for Problem 2 into a lower bound for Problem 1.
In fact, we prove the stronger bound on the \emph{minimum} max load for Problem 1 over all choices of random seed (as opposed to the expected max load).

\begin{theorem}
\label{thm:problem2-to-problem1}
Let \(1\leq n\leq u<p\), where \(p\) is prime, and let \(U=[u]\).
Then for every set \(X\subseteq U\),
\begin{equation}
\mathbb{E}_{\hmod_s \sim \mathcal{H}^n_p} [\maxload_X(\hmod_s)] \ge \min_{s \in \mathbb{Z}_p} \left(\maxload_X(\hmod_s)\right) \ge \min_{a \in \mathbb{T}} \left(\maxload_X(\hreal_a)\right).
\label{eq:problem2-to-problem1}
\end{equation}
\end{theorem}

We will describe in Section \ref{sec:uppers} how to extend Theorem \ref{thm:problem2-to-problem1} to show that, when $p$ is large enough, the \emph{expected} max loads for Problems 1 and 2 are essentially the same.
This will allow us to establish a connection between the known upper bounds for the two problems (namely, the known upper bound for Problem 2 \cite{konyagin-ruzsa-schlag2000} actually implies an $n^{1/3 + o(1)}$ upper bound for Problem 1, which is only slightly weaker than the state of the art upper bound for Problem 1 given in \cite{knudsen2016linear,knudsen2019linear}).

We remark that the proof of Theorem \ref{thm:problem2-to-problem1} (as well as the extended version of the reduction in Section \ref{sec:uppers}) can be viewed as a strengthening of a reduction that appeared in an earlier preprint \cite{westover2023relationship} by one of the authors of this paper.

\begin{proof}[Proof of Theorem \ref{thm:problem2-to-problem1}]
Because \(\gcd(p,n)=1\), there exists \(c\in[n]\) such that \(pc\equiv-1\pmod n\).
Define, for \(s\in[p]\),
\[
 a_s\coloneqq\left\{\frac{s(c+1/p)}{n}\right\}\in\mathbb T.
\]

We first prove the following. 

\begin{lemma}\label{lem:thm2to1-helper}
For every \(x\in U\) and every $s \in [p]$,
\begin{equation}
 \imod{c\,\hmod_s(x)}{n}=\hreal_{a_s}(x).
 \label{eq:problem2-to-problem1-2}
\end{equation}
\end{lemma}

\begin{proof} 
For every real \(\alpha\) and integer \(x\),
\[
 \hreal_{\{\alpha\}}(x)
 = \left\lfloor n\{\alpha x\}\right\rfloor
 =\imod{\lfloor n\alpha x\rfloor}{n},
\]
where the latter equality holds since $\lfloor n \alpha x \rfloor =n\lfloor \alpha x \rfloor + \lfloor n \{\alpha x\} \rfloor$.
Applying this with \(\alpha=s(c+1/p)/n\),
and since \(csx\) is an integer, we have 
\[
 \hreal_{a_s}(x)
 =\imod{csx+\lfloor sx/p\rfloor}{n}.
\]
On the other hand, \(\imod{sx}{p}=sx-p\lfloor sx/p\rfloor\), and \(-pc\equiv1\pmod n\).
Consequently,
\begin{align*}
 \imod{c\,\hmod_s(x)}{n}
  =\imod{c\,\imod{\imod{sx}{p}}{n}}{n}
 =\imod{c\,\imod{sx}{p}}{n}
 = \imod{csx-pc\lfloor sx/p\rfloor}{n}
 =\imod{csx+\lfloor sx/p\rfloor}{n}
 =\hreal_{a_s}(x).
\end{align*}
\end{proof}

With \cref{lem:thm2to1-helper}, it remains to show \eqref{eq:problem2-to-problem1}.
Note that, since \(pc\equiv-1\pmod n\), we are guaranteed that \(\gcd(c,n)=1\), and thus that multiplication by \(c\) is a bijection on \(\mathbb{Z}_n\). Thus,
\[
\maxload_X(\hmod_s) = \maxload_X(c \hmod_s).
\]
Together with \eqref{eq:problem2-to-problem1-2} this shows
\[
\maxload_X(\hmod_s) = \maxload_X(\hreal_{a_s}),
\]
and the theorem follows.
\end{proof}

We remark that this first reduction is already enough to obtain a non-trivial lower bound for Problem 1. Indeed, the lower bound in Theorem \ref{thm:krs-bounds} combined with Theorem \ref{thm:problem2-to-problem1} immediately implies:

\begin{corollary}[A first lower bound for Problem 1]
\label{cor:first-lower-bound}
For every sufficiently large \(n\), there exists a universe size \(u=n^{1+o(1)}\) such that, for every prime \(p>u\), there exists a set \(X\subseteq[u]\) of size \(n\) such that
    \[
\mathbb{E}_{\hmod_s \sim \mathcal{H}^n_p} [\maxload_X(\hmod_s)] \ge \min_{s \in \mathbb{Z}_p} \left(\maxload_X(\hmod_s)\right) \ge \exp(\Omega(\log n / (\log \log n)^2)).
\]
\end{corollary}

\paragraph{Reduction 2: From arithmetic Kakeya to real linear hashing.}
Our second reduction shows that small arithmetic Kakeya sets immediately imply lower bounds for Problem 2.
Specifically, the existence of an arithmetic Kakeya set $A$ gives a lower bound on the max load for \emph{any} $n$ element set containing $A$ with \emph{any} random seed.
\begin{theorem}
\label{thm:progressions-force-load}
Let \(A\) be a set of at most \(n\) integers containing a \(k\)-term arithmetic progression with every positive common difference \(d\in[D]\).
If \(D\geq kn\), then every \(n\)-element set \(X\supseteq A\) and every \(a\in\mathbb T\) satisfy
\[
 \maxload_X(\hreal_a)\geq\frac{k}{2}.
\]
\end{theorem}

\begin{proof}
The claim is immediate when \(k=1\), so assume \(k\geq2\).
Fix \(X\supseteq A\) and \(a\in\mathbb T\).
Among the \(D\) points \(ra\), for \(r\in[D]\), there must exist two points \(r_1a\) and \(r_2a\), with \(r_1<r_2\), whose circular distance is at most \(1/D\).
Subtracting their indices gives a positive \(d=r_2-r_1\in[D]\) such that
\[
 \|ad\|_{\mathbb T}
 \coloneqq
 \min_{z\in\mathbb Z}|ad-z|
 \leq
 \frac{1}{D}.
\]
By assumption, $A$ contains a $k$-term arithmetic progression with difference $d$, write this as \(x,x+d,\ldots,x+(k-1)d\).
Multiplying the progression by \(a\) and examining the results modulo \(1\) gives a sequence \(\{ax\},\{ax+ad\},\ldots,\{ax+(k-1)ad\}\in\mathbb T\) that lies entirely in a circular interval of length at most
\[
 (k-1)\|ad\|_{\mathbb T}
 \leq
 \frac{k-1}{D}
 <
 \frac{1}{n}.
\]
Such an interval meets at most two of the \(n\) hash bins, so one of those bins contains at least \(k/2\) terms of the progression.
\end{proof}

\paragraph{The final lower bounds. }Finally, we can feed the arithmetic Kakeya sets of \cref{thm:green-ruzsa} through our two reductions in order to extract lower bounds of \(\exp(\Omega(\log n/\log\log n))\) for each of Problems 2 and 1.
For Problem 2, this improves the lower bound of \(\exp(\Omega(\log n/(\log\log n)^2))\) from \cref{thm:krs-bounds}; for Problem 1, it improves \cref{cor:first-lower-bound}, and therefore also the \(\Omega(\log n/\log\log n)\) bound that holds for every hash family.

\begin{corollary}[Lower bound for Problem 2]
\label{cor:problem2-lower-bound}
For every sufficiently large \(n\), there exists a universe size \(u = n^{1 + o(1)}\) and a set \(X \subseteq [u]\) of size \(n\) such that every \(a \in \mathbb{T}\) satisfies
\[
 \maxload_X(\hreal_a)
 \geq
 \exp\left(\Omega\left(\frac{\log n}{\log\log n}\right)\right).
\]
In particular, for this choice of \(u\), the quantity \(K(n, u)\) given by \eqref{eq:krs-quantity} (the best-case max load over all seeds) satisfies \(K(n, u) \geq \exp(\Omega(\log n / \log\log n))\).
\end{corollary}

\begin{proof}
Let \(n\) be sufficiently large, let \(K = \exp(\Theta(\log n / \log\log n))\), and let \(X \subseteq [u]\) with \(u = n^{1 + o(1)}\) be the key set supplied by \cref{thm:green-ruzsa}.
That is, \(|X|=n\), and \(X\) contains a \(K\)-term arithmetic progression with every positive common difference \(d\in[nK]\).
Applying \cref{thm:progressions-force-load} with \(A = X\), with \(k = K\), and with \(D = nK\) gives the claimed result.
\end{proof}

\begin{corollary}[Lower bound for Problem 1]
\label{cor:problem1-lower-bound}
For every sufficiently large \(n\), there exists a universe size \(u = n^{1 + o(1)}\) and a set \(X \subseteq [u]\) of size \(n\) such that, for every prime \(p > u\),
\[
    \mathbb{E}_{\hmod_s \sim \mathcal{H}^n_p}[\maxload_X(\hmod_s)] \ge \min_{s \in \mathbb{Z}_p} \left(\maxload_X(\hmod_s)\right)
 \geq
 \exp\left(\Omega\left(\frac{\log n}{\log\log n}\right)\right).
\]
\end{corollary}

\begin{proof}
Let \(u\) and \(X\) be as in \cref{cor:problem2-lower-bound}, and let \(p > u\) be prime.
Since \(X \subseteq [u]\) consists of \(n\) distinct elements, we have \(u \geq n\) and hence \(n < p\), so \cref{thm:problem2-to-problem1} applies to \(X\) and gives
\[
 \mathbb{E}_{\hmod_s \sim \mathcal{H}^n_p}[\maxload_X(\hmod_s)]
 \ge \min_{s \in \mathbb{Z}_p}\left(\maxload_X(\hmod_s)\right)
 \ge \min_{a \in \mathbb{T}}\left(\maxload_X(\hreal_a)\right).
\]
By \cref{cor:problem2-lower-bound}, the right-hand side is at least \(\exp(\Omega(\log n / \log\log n))\).
\end{proof}

One surprising feature of our lower bound for Problem 1 is that it holds not just in expectation (for a random seed \(s\)), but for every seed \(s\). That is, for every $s \in \mathbb{Z}_p$, the maximum load of $\hmod_s$ is at least $\exp(\Omega(\log n / \log\log n))$. Thus modular linear hashing is exponentially worse than fully random hashing, whose maximum load is \(\Theta(\log n / \log\log n)\) with high probability.

\section{A Stronger Connection Between Problems 1 and 2}
\label{sec:uppers}

We now extend the connection between Problems 1 and 2 to show that, when $p$ is large enough, the expected max loads for modular and real linear hashing are essentially the same.

\begin{theorem}
\label{thm:stronger-connection}
Let \(0 < n \le u\), and let \(p > n^2 u\) be prime.
Then
\[
 \left|
 \max_{X\subseteq U,\,|X|=n}
 \mathbb E_{\hmod_s\sim\mathcal H^n_p}\left[\maxload_X(\hmod_s)\right]
 -
 \max_{X\subseteq U,\,|X|=n}
 \mathbb E_{\hreal_a\sim\mathcal H^n_{\mathbb T}}\left[\maxload_X(\hreal_a)\right]
 \right|
 \leq
 1.
\]
\end{theorem}

The rest of the section proves \cref{thm:stronger-connection}, in two steps.
The proof of \Cref{thm:problem2-to-problem1} already pairs each \(s\in[p]\) with a value \(a_s\in\mathbb T\) for which \(\hmod_s\) and \(\hreal_{a_s}\) induce the same max load, so the first step is to determine which \(a\in\mathbb T\) arise as some \(a_s\).
The next lemma shows that, as \(s\) ranges over \([p]\), the value \(a_s\) ranges over an evenly spaced grid of \(p\) points in \(\mathbb T\); the expected modular max load is therefore exactly the average of the real max load over this grid.
The second step is to bound the difference between that grid average and the average of the real max load over all of \(\mathbb T\).

\begin{lemma}
\label{lem:a-s-grid}
Let \(c\in [n]\) such that \(pc\equiv-1\pmod n\), and for \(s\in[p]\), let
$a_s=\left\{\frac{s(c+1/p)}{n}\right\}$.
Then \(s\mapsto a_s\) is a bijection from \([p]\) to the \(p\)-point grid
\[
 \left\{\frac{j}{p}:j\in[p]\right\}\subset\mathbb T.
\]
In particular, a uniformly random $s$ induces the uniform distribution on this grid.
\end{lemma}

\begin{proof}
Let \(t = (pc+1)/n\in\mathbb Z\).
Then
\[
 a_s
 =\left\{\frac{s(c+1/p)}{n}\right\}
 =\left\{\frac{st}{p}\right\}
 =\frac{\imod{st}{p}}{p}.
\]
Since \(nt=pc+1\equiv1\pmod p\), we have \(\gcd(t,p)=1\).
Thus multiplication by \(t\) is a bijection on the integers modulo \(p\), and \(s\mapsto\imod{st}{p}\) is a permutation of \([p]\).
\end{proof}

With \cref{lem:a-s-grid} in hand, we can prove the following connection between real and modular expected max loads.

\begin{proposition}
\label{prop:expected-maxload}
Let $0 < n \le u < p$, where $p$ is a prime. Fix a key set \(X\subseteq U\), and let \(m=|X|\).
Let \(S\) be uniform on \([p]\), and let \(A\) be uniform on \(\mathbb T\).
Then
\[
 \left|
 \mathbb E_S\!\left[\maxload_X(\hmod_S)\right]
 -
 \mathbb E_A\!\left[\maxload_X(\hreal_A)\right]
 \right|
 \leq
 \frac{n}{p}\sum_{x\in X}x
 \leq
 \frac{n(u-1)m}{p}.
\]
\end{proposition}

\begin{proof}
Define
\[
 L_X(a)\coloneqq\maxload_X(\hreal_a),
 \qquad a\in[0,1).
\]
For every \(s\in[p]\), the proof of \cref{thm:problem2-to-problem1} shows that \(\hreal_{a_s}\) agrees with \(\hmod_s\) up to a permutation of the bins, so \(\maxload_X(\hmod_s)=L_X(a_s)\).
By \cref{lem:a-s-grid}, the values \(a_s\) range bijectively over \(\{j/p:j\in[p]\}\), so averaging over \(s\) gives
\begin{equation}
 \mathbb E_S\!\left[\maxload_X(\hmod_S)\right]
 =
 \frac{1}{p}\sum_{j=0}^{p-1}L_X\!\left(\frac{j}{p}\right).
 \label{eq:grid-average}
\end{equation}
We call the right-hand side of \eqref{eq:grid-average} the \defn{grid average} of \(L_X\).
Since \(A\) is uniform on \(\mathbb T\), we likewise have
\begin{equation}
 \mathbb E_A\!\left[\maxload_X(\hreal_A)\right]
 =
 \int_0^1 L_X(a)\,da,
 \label{eq:continuous-average}
\end{equation}
and we call the right-hand side of \eqref{eq:continuous-average} the \defn{continuous average} of \(L_X\).
It remains to bound the difference between \eqref{eq:grid-average} and \eqref{eq:continuous-average}.

Recall that the \defn{total variation} of a step function is the sum of the magnitudes of its jumps.
The function \(L_X\) is a right-continuous step function.
For each \(x\in X\), the bin \(\hreal_a(x)\) changes only when \(a=k/(nx)\) for an integer \(k\) with \(1\leq k<nx\).
If \(r\) keys change bins at the same value of \(a\), then the maximum load changes by at most \(r\).
Hence the total variation of \(L_X\) is at most
\[
 n\sum_{x\in X}x.
\]

We claim that, for any right-continuous step function of total variation \(V\), the error of its left Riemann sum on a mesh of width \(1/p\) is at most \(V/p\).
To see this, let \(f\) be such a function on \([0,1)\), and let \(I_j\coloneqq[j/p,(j+1)/p)\) for \(j\in[p]\) denote the mesh intervals, so that the left Riemann sum in question is \(\frac{1}{p}\sum_{j=0}^{p-1}f(j/p)\).
Write \(V_j\) for the sum of the magnitudes of the jumps of \(f\) that occur inside \(I_j\).
Since \(f\) is a right-continuous step function, its value at any \(a\in I_j\) is obtained from its value at the left endpoint \(j/p\) by adding the jumps that occur in \((j/p,a]\).
In particular, \(|f(a)-f(j/p)|\leq V_j\) for every \(a\in I_j\), so integrating over \(I_j\), which has width \(1/p\), gives
\[
 \left|
 \int_{I_j}f(a)\,da
 -
 \frac{1}{p}f\!\left(\frac{j}{p}\right)
 \right|
 \leq
 \frac{V_j}{p}.
\]
The mesh intervals partition \([0,1)\), so each jump of \(f\) is counted in exactly one \(V_j\), meaning that \(\sum_{j}V_j=V\).
Summing the previous display over \(j\) therefore bounds the total error by \(V/p\), which proves the claim.

The grid average \eqref{eq:grid-average} is precisely the left Riemann sum of \(L_X\) on the mesh of width \(1/p\), so the claim yields
\[
 \left|
 \frac{1}{p}\sum_{j=0}^{p-1}L_X\!\left(\frac{j}{p}\right)
 -
 \int_0^1L_X(a)\,da
 \right|
 \leq
 \frac{n}{p}\sum_{x\in X}x.
\]
Finally, \(x\leq u-1\) for every \(x\in X\), so \(\sum_{x\in X}x\leq(u-1)m\).
\end{proof}

\begin{proof}[Proof of \cref{thm:stronger-connection}]
Since \(p>n^2u\), applying \cref{prop:expected-maxload} with \(m=n\) shows that the two expectations differ by at most \(n(u-1)n/p<1\) for every \(X\subseteq U\) of size \(n\).
The same is therefore true of their maxima over \(X\).
\end{proof}

\subsection{Consequence: An Alternative \texorpdfstring{\(n^{1/3+o(1)}\)}{n to the 1/3+o(1)} Upper Bound for Problem 1}
Because \cref{thm:stronger-connection} compares the two problems in both directions, it also lets us deduce upper bounds for Problem 1 from upper bounds for Problem 2.
The upper bound recorded in \cref{thm:krs-bounds} concerns the best case over all \(a\in\mathbb T\), and as stated it says nothing about a uniformly random \(a\).
The argument of Konyagin, Ruzsa, and Schlag, using estimates from work of Alon--Peres and Alon--Ruzsa, already bounds the max load \emph{averaged} over a uniformly random \(a\in\mathbb T\) \cite{konyagin-ruzsa-schlag2000,alon-peres1992uniform,alon-ruzsa1999nonaveraging}.

\begin{proposition}[\cite{konyagin-ruzsa-schlag2000,alon-peres1992uniform,alon-ruzsa1999nonaveraging}]
\label{prop:krs-average-upper}
Let \(0<n\le u\), and let \(X\subseteq U\) be a key set of size \(n\).
Then
\[
 \mathbb E_{\hreal_a\sim\mathcal H^n_{\mathbb T}}\left[\maxload_X(\hreal_a)\right]
 \leq
 n^{1/3}\exp\left(O\left(\sqrt{\log n\log\log n}\right)\right).
\]
\end{proposition}

\begin{proof}
Write \(X=\{x_0,\ldots,x_{n-1}\}\) and, for distinct \(i,j,k\in[n]\), let
\[
 Z_{ijk}(a)
 \coloneqq
 \mathbf 1\left[
   \|a(x_i-x_j)\|_{\mathbb T}<\frac{1}{n}
   \ \text{ and }\
   \|a(x_i-x_k)\|_{\mathbb T}<\frac{1}{n}
 \right],
\]
where \(\|t\|_{\mathbb T}\) denotes the distance from \(t\) to the nearest integer, and let \(Z(a)\) be the sum of \(Z_{ijk}(a)\) over all triples of distinct \(i,j,k\in[n]\).
The third-moment estimate at the heart of the Konyagin--Ruzsa--Schlag proof \cite[proof of Theorem 5.2]{konyagin-ruzsa-schlag2000}, with the estimates from Alon--Peres and Alon--Ruzsa \cite{alon-peres1992uniform,alon-ruzsa1999nonaveraging} retained rather than absorbed into an \(n^{o(1)}\) factor, bounds the expected number of such triples by
\[
 \mathbb E_a\left[Z(a)\right]
 \leq
 n\exp\left(O\left(\sqrt{\log n\log\log n}\right)\right).
\]
On the other hand, if some bin of \(\hreal_a\) contains \(L\) keys, then those keys all lie in a common interval of length \(1/n\), so each of the \(L(L-1)(L-2)\) ordered triples of distinct keys drawn from that bin contributes a term to \(Z(a)\).
Applying this to a fullest bin, and abbreviating \(M(a)\coloneqq\maxload_X(\hreal_a)\), we get \(Z(a)\geq M(a)(M(a)-1)(M(a)-2)\), and hence \(M(a)^3=O(Z(a)+1)\), where the additive constant absorbs the case \(M(a)\leq3\).
Taking expectations and applying Jensen's inequality in the form \(\mathbb E[Y]\leq(\mathbb E[Y^3])^{1/3}\) now gives
\[
 \mathbb E_a\left[\maxload_X(\hreal_a)\right]
 \leq
 \left(\mathbb E_a\left[M(a)^3\right]\right)^{1/3}
 \leq
 n^{1/3}\exp\left(O\left(\sqrt{\log n\log\log n}\right)\right).\qedhere
\]
\end{proof}

Combining \cref{prop:krs-average-upper} with \cref{thm:stronger-connection}, we are able to obtain an $n^{1/3 + o(1)}$ upper bound for Problem 1 directly from known results for Problem 2 as \Cref{cor:modular-upper-from-krs}.
However, we note that \Cref{cor:modular-upper-from-krs} is slightly weaker than Knudsen's \(O((n\log n)^{1/3})\) bound from \cref{thm:knudsen-upper-bound}, and it applies only when \(p\) is large relative to \(n\) and \(u\).

\begin{corollary}
\label{cor:modular-upper-from-krs}
Let \(0<n\le u\), and let \(p>n^2u\) be prime.
Then
\[
 \max_{X\subseteq U,\,|X|=n}
 \mathbb E_{\hmod_s\sim\mathcal H^n_p}\left[\maxload_X(\hmod_s)\right]
 \leq
 n^{1/3}\exp\left(O\left(\sqrt{\log n\log\log n}\right)\right)
 =
 n^{1/3+o(1)}.
\]
\end{corollary}

\begin{proof}
By \cref{prop:krs-average-upper}, the expected real max load is at most \(n^{1/3}\exp(O(\sqrt{\log n\log\log n}))\) for every \(X\subseteq U\) of size \(n\), and by \cref{thm:stronger-connection} the corresponding modular maximum exceeds this by at most \(1\).
\end{proof}

\section{Arithmetic Kakeya and Linear Hashing}
\label{sec:hashing-implies-kakeya}

The reduction in \cref{thm:progressions-force-load} identifies arithmetic Kakeya sets as a source of pointwise lower bounds for linear hashing.
In this section we show that bounds for (real and modular) linear hashing are naturally connected to a polynomial-length version of the arithmetic Kakeya conjecture (\Cref{conj:arithmetic-kakeya}).

\begin{conjecture}[Polynomial-length arithmetic Kakeya]
\label{conj:polynomial-arithmetic-kakeya}
For positive integers \(k\) and \(D\), let \(F_k(D)\) be the minimum size of a finite set of integers containing a \(k\)-term arithmetic progression with every positive common difference \(d\in[D]\).
For every constant \(\eta>0\),
\[
 \liminf_{D\to\infty}
 \frac{\log F_{\lfloor D^\eta\rfloor}(D)}{\log D}
 \geq
 1.
\]
Equivalently, for every \(\eta,\varepsilon>0\) and all sufficiently large \(D\),
\[
 F_{\lfloor D^\eta\rfloor}(D)
 \geq
 D^{1-\varepsilon}.
\]
\end{conjecture}

Since \(F_k(D)\) is nondecreasing in \(k\), the usual arithmetic Kakeya conjecture (\Cref{conj:arithmetic-kakeya}) implies \cref{conj:polynomial-arithmetic-kakeya}.
Although this polynomial-length form is weaker than the usual arithmetic Kakeya conjecture, it is known to imply the Kakeya conjecture for upper Minkowski dimension \cite{bourgain1991remarks,bourgain1993distribution}.
The connection between arithmetic projections and the Kakeya problem was developed further by Bourgain, Katz, and Tao \cite{bourgain1999dimension,katz-tao1999arithmetic-projections,katz2002newbounds}; Green and Ruzsa give the precise progression formulation used here \cite[Section 1, Equation (1.1)]{green-ruzsa2019arithmetic}.

\subsection{The Polynomial-Length Arithmetic Kakeya Conjecture Limits Lower Bounds}
\label{sec:kakeya-barrier}

The arithmetic Kakeya sets of \cref{thm:green-ruzsa} are the source of our strongest lower bounds, so it is natural to ask how much further this approach can be pushed.
This subsection observes that \cref{conj:polynomial-arithmetic-kakeya} caps it strictly below any polynomial bound.

Any quantitative improvement to the Green--Ruzsa upper bound on \(F_k(D)\) would yield a stronger hashing lower bound through \cref{thm:progressions-force-load}.
To make this precise, given \(k\), let \(\delta_k>0\) be such that
\[
 F_k(D)\leq D^{1-\delta_k}
\]
holds for all sufficiently large \(D\).
Green--Ruzsa provides such a bound with \(\delta_k=\Omega(1/\log\log k)\).
When \(n\) is sufficiently large, setting \(D=kn\) gives a set \(A\) of at most \((kn)^{1-\delta_k}\) integers that contains a \(k\)-term arithmetic progression with every positive common difference \(d\in[kn]\).
For sufficiently large \(n\) it holds that \((kn)^{1-\delta_k}\leq n\), or equivalently that \( k^{1-\delta_k} \leq n^{\delta_k}\).
Therefore, the set \(A\) has at most \(n\) elements, and \cref{thm:progressions-force-load} forces a max load of at least \(k/2\) on every \(n\)-element set of keys containing \(A\).
A larger \(\delta_k\) therefore admits a larger progression length \(k\), and hence a stronger lower bound.
\Cref{conj:polynomial-arithmetic-kakeya} says that if \(k\) is polynomial in \(n\), then \(\delta_k\) tends to 0.
This would rule out reaching a polynomial lower bound through this specific type of construction.

\begin{remark}
\label{prop:arithmetic-kakeya-barrier}
Assume \Cref{conj:polynomial-arithmetic-kakeya}.
The arithmetic-Kakeya-set construction in \cref{thm:progressions-force-load} cannot prove an \(n^{\Omega(1)}\) maximum-load lower bound for all sufficiently large \(n\).
\end{remark}


\subsection{Subpolynomial Linear Hashing Upper Bounds Imply Polynomial-Length Arithmetic Kakeya}

The best known upper bounds for modular and real linear hashing are polynomial in the number of keys \cite{konyagin-ruzsa-schlag2000,knudsen2016linear,knudsen2019linear}.
We show that improving the upper bound to \(n^{o(1)}\)---even for a single seed---would establish \cref{conj:polynomial-arithmetic-kakeya}.

Before proving the connection to linear hashing, we first show that the location of a progression-rich set can be normalized without losing its progressions.

\begin{lemma}
\label{lem:pack-arithmetic-kakeya-set}
Suppose that a finite set \(A\subseteq\mathbb Z\) contains a \(k\)-term arithmetic progression with every positive common difference \(d\in[D]\).
Then there is a set \(A'\subseteq[D|A|]\) with \(|A'|=|A|\) having the same property.
\end{lemma}

\begin{proof}
Choose one such progression for each difference \(d\) (truncating to exactly \(k\) terms if necessary) and form a weighted graph on \(A\) by adding an edge with length \(d\) between consecutive terms with difference \(d\).
Every edge has length less than \(D\), so a connected component \(C\) has diameter at most \((D-1)(|C|-1)<D|C|\).
Translate the connected components independently into disjoint consecutive intervals, assigning an interval of length \(D|C|\) to each component.
Every chosen progression lies in one component and is preserved by its translation, while the total length of the intervals is \(D|A|\).
\end{proof}

We now show that even a weak subpolynomial upper bound on the maximum load of either real or modular hashing is enough to imply \cref{conj:polynomial-arithmetic-kakeya}.
In particular, it is sufficient if for every set there exists even a single seed with subpolynomial max load, which we refer to as a \defn{single-seed} upper bound.

\begin{theorem}
\label{thm:subpolynomial-hashing-implies-kakeya}
Let \(H(n)=n^{o(1)}\).
Suppose that, for all sufficiently large \(n\), all \(u > n\) with \(u=n^{O(1)}\), and all \(n\)-element key sets \(X\subseteq[u]\),
\[
 \min_{a\in\mathbb T}\maxload_X(\hreal_a)
 \leq
 H(n).
\]
Then \cref{conj:polynomial-arithmetic-kakeya} holds.
\end{theorem}

\begin{proof}
Fix \(\eta>0\), and suppose for the sake of contradiction that \cref{conj:polynomial-arithmetic-kakeya} fails for this \(\eta\).
Negating the equivalent formulation in \cref{conj:polynomial-arithmetic-kakeya}, there is a constant \(\varepsilon>0\) such that, for arbitrarily large integers \(D\),
\[
 F_{\lfloor D^\eta\rfloor}(D)
 \leq
 D^{1-\varepsilon}.
\]
Set
\[
 \delta
 \coloneqq
 \frac{1}{2}\min\{\eta,\varepsilon,1\},
 \qquad
 k
 \coloneqq
 \lfloor D^\delta\rfloor,
 \qquad
 n
 \coloneqq
 \left\lfloor\frac{D}{k}\right\rfloor.
\]
For all sufficiently large such \(D\), we have \(2\leq k\leq \lfloor D^\eta\rfloor\),
\[
 F_k(D)
 \leq
 F_{\lfloor D^\eta\rfloor}(D)
 \leq
 D^{1-\varepsilon}
 =
 o(n),
\]
and \(D\geq kn\).

By the definition of \(F_k(D)\), there is a finite set \(A\subseteq\mathbb Z\) with \(|A|\leq D^{1-\varepsilon}\) that contains a \(k\)-term arithmetic progression with every positive common difference \(d\in[D]\).
\Cref{lem:pack-arithmetic-kakeya-set} lets us assume \(A\subseteq[D|A|]\).
Pad \(A\) to an \(n\)-element set \(X\subseteq[u]\), where \(u\coloneqq D|A|+n+1>n\).
Since \(n=\Theta(D^{1-\delta})\),
\[
 u
 \leq
 D^{2-\varepsilon}+n+1
 =
 n^{O(1)}.
\]
Applying \cref{thm:progressions-force-load} gives
\[
 \maxload_X(\hreal_a)
 \geq
 \frac{k}{2}
\]
for every \(a\in\mathbb T\).

Moreover,
\[
 \frac{\log k}{\log n}
 \longrightarrow
 \frac{\delta}{1-\delta}
 >
 0,
\]
so \(k/2=n^{\Omega(1)}\), contradicting the assumed minimum maximum-load bound \(H(n)=n^{o(1)}\).
\end{proof}

Because an expected upper bound implies a single-seed upper bound, we can immediately conclude that either type of upper bound for real hashing implies \cref{conj:polynomial-arithmetic-kakeya}.
Similarly, using \cref{thm:problem2-to-problem1}, we conclude the same for modular linear hashing.

\begin{corollary}
\label{cor:expected-modular-hashing-implies-kakeya}
Let \(H(n)=n^{o(1)}\).
Suppose that there is a uniform single-seed or expected max-load upper bound of \(H(n)\) for either real linear hashing on universes \([u]\) with \(u=n^{O(1)}\), or modular linear hashing with \(u<p\) and \(u,p=n^{O(1)}\).
Then \cref{conj:polynomial-arithmetic-kakeya} holds.
\end{corollary}

Using an arithmetic-projection argument, Bourgain showed that \Cref{conj:polynomial-arithmetic-kakeya} implies the Kakeya conjecture for upper Minkowski dimension \cite{bourgain1991remarks,bourgain1993distribution}; see also the discussion of Equation (1.1) by Green and Ruzsa \cite[Section 1]{green-ruzsa2019arithmetic} and the subsequent projection bounds of Katz and Tao \cite{katz-tao1999arithmetic-projections,katz2002newbounds}.
A Besicovitch set in \(\mathbb R^d\) is a compact measure-zero set containing a unit line segment in every direction.
For a bounded set \(E\subseteq\mathbb R^d\), if \(N(E,r)\) is the minimum number of radius-\(r\) balls needed to cover \(E\), then its upper Minkowski dimension is
\[
 \overline{\dim}_{\mathrm M}(E)
 \coloneqq
 \limsup_{r\to0}
 \frac{\log N(E,r)}{\log(1/r)}.
\]
The Kakeya conjecture for upper Minkowski dimension states that every Besicovitch set in \(\mathbb{R}^d\) has upper Minkowski dimension \(d\) for every \(d \ge 2\).

\begin{corollary}
\label{cor:subpolynomial-hashing-implies-upper-minkowski-kakeya}
The existence of a uniform single-seed or expected max-load upper bound of \(n^{o(1)}\) for either real linear hashing on universes \([u]\) with \(u=n^{O(1)}\), or modular linear hashing with \(u<p\) and \(u,p=n^{O(1)}\) implies that every Besicovitch set in \(\mathbb R^d\) has upper Minkowski dimension \(d\), for every \(d\geq2\).
\end{corollary}

\paragraph{Why the full arithmetic Kakeya conjecture does not follow.}
The usual arithmetic Kakeya conjecture first fixes the progression length \(k\) and then sends \(D\) to infinity, whereas \cref{conj:polynomial-arithmetic-kakeya} couples the two parameters by taking \(k=\lfloor D^\eta\rfloor\).
Indeed, suppose for some fixed \(k\) and arbitrarily large \(D\) that \(F_k(D)\leq D^{1-\varepsilon}\).
Repeating the proof of \cref{thm:subpolynomial-hashing-implies-kakeya} with \(n\coloneqq\lfloor D/k\rfloor\) would produce an \(n\)-element key set whose maximum load is at least \(k/2\) for every seed.
Since \(n=\Theta(D)\) while \(k/2\) remains constant, this contradicts a hashing upper bound \(H(n)\) only if \(H(n)<k/2\) for arbitrarily large \(n\).
Thus, through the present reduction, hashing upper bounds can imply \Cref{conj:polynomial-arithmetic-kakeya} and hence the Kakeya conjecture for upper Minkowski dimension, but not the full arithmetic Kakeya conjecture.

\section{Acknowledgments}

This work was supported in part by NSF grants CCF-2542165, CNS-2504471, and CNS-2504470, and by a Jane Street grant. This work was done in part while Hanna Koml\'os was visiting the Simons Institute for the Theory of Computing, and was supported in part by Google Research, Fall 2025.

Opus 4.8, Opus 5, GPT 5.5, and GPT 5.6 were all used to edit and create initial drafts of portions of this paper, which were then heavily edited by the authors. Additionally, the proof of \Cref{thm:progressions-force-load} was developed in collaboration with GPT 5.6. Finally, GPT 5.6 was used to find papers on Problems 2 and 3.

\printbibliography

\appendix
\crefalias{section}{appendix}
\section{The Green--Ruzsa Construction}
\label{app:arithmetic-kakeya}

This appendix supplies the proof of \cref{thm:green-ruzsa}, the arithmetic Kakeya construction on which the lower bounds of \cref{cor:problem2-lower-bound,cor:problem1-lower-bound} rest.
The construction itself is due to Green and Ruzsa \cite{green-ruzsa2019arithmetic}; what we must supply is a version of their result whose parameters are tuned to a prescribed number of keys \(n\).

Green and Ruzsa prove
\[
 \lim_{D\to\infty}
 \frac{\log F_k(D)}{\log D}
 \leq
 1-\frac{c}{\log\log k}
\]
for an absolute constant \(c>0\) \cite[Theorem 1.2]{green-ruzsa2019arithmetic}.
This asymptotic statement alone is not sufficient for our application because its threshold in \(D\) need not be uniform in \(k\).
Instead, we extract the following explicit estimate from their proof, retaining the number of prime factors as a free parameter.

\begin{lemma}[Green--Ruzsa {\cite[Section 5]{green-ruzsa2019arithmetic}}]
\label{lem:green-ruzsa-construction}
Let \(p_1,\ldots,p_m\) be distinct odd primes, let \(D\coloneqq\prod_{i=1}^m p_i\), and let \(k\geq2\).
There is a set \(S\subseteq[kD]\) that contains a \(k\)-term arithmetic progression with every positive common difference \(d\in[D]\) and satisfies
\[
 |S|
 \leq
 k^2 2^{-m}D
 \prod_{i=1}^m\left(1+\frac{1}{p_i}\right).
\]
\end{lemma}

\begin{proof}
For each positive \(d\in[D]\), let \(x_d\in[D]\) be the unique integer satisfying \(x_d\equiv d^2\pmod D\), and set
\[
 S\coloneqq\{x_d+jd:d\in[D]\setminus\{0\},\ j\in[k]\}.
\]
For each \(d\), the corresponding \(k\) elements form an arithmetic progression with common difference \(d\), and they lie in \([kD]\) because \(0\leq x_d+jd<kD\).

Fix \(j\in[k]\).
For every \(i\),
\[
 4(x_d+jd)
 \equiv
 4(d^2+jd)
 =
 (2d+j)^2-j^2
 \pmod{p_i}.
\]
As a function of \(d\in\mathbb Z_{p_i}\), the right-hand side takes exactly \((p_i+1)/2\) values.
The Chinese remainder theorem therefore shows that \(x_d+jd\bmod D\) takes at most
\[
 \prod_{i=1}^m\frac{p_i+1}{2}
 =
 2^{-m}D\prod_{i=1}^m\left(1+\frac{1}{p_i}\right)
\]
values.
Because \(0\leq x_d+jd<kD\), each residue class modulo \(D\) has at most \(k\) representatives among the possible integer values of \(x_d+jd\).
Thus, for each fixed \(j\),
\[
 \bigl|\{x_d+jd:d\in[D]\setminus\{0\}\}\bigr|
 \leq
 k 2^{-m}D\prod_{i=1}^m\left(1+\frac{1}{p_i}\right).
\]
Summing this bound over the \(k\) choices of \(j\) proves the result.
\end{proof}

We tune the number of prime factors in the Green--Ruzsa construction to the desired number of keys.
This avoids the loss incurred by first fixing the progression length and then tensoring the construction to an arbitrary scale.
We now restate \cref{thm:green-ruzsa}; beyond what the statement requires, the proof below also tracks the modulus \(D\) and an upper bound on \(K\).

\greenruzsatheorem*

\begin{proof}
Let \(3=p_1<p_2<\cdots\) be the odd primes, and write
\[
 D_m\coloneqq\prod_{i=1}^m p_i,
 \qquad
 K_m\coloneqq\left\lfloor\left(\frac{3}{2}\right)^{m/5}\right\rfloor.
\]
Choose \(m\) minimally so that \(D_m>nK_m\), and abbreviate \(D\coloneqq D_m\) and \(K\coloneqq K_m\).
Such an \(m\) exists and tends to infinity with \(n\).
By minimality,
\[
 D
 =
 p_mD_{m-1}
 \leq
 p_mnK_{m-1}
 \leq
 p_mnK.
\]

Apply \cref{lem:green-ruzsa-construction} with these parameters.
Since every \(p_i\geq3\),
\[
 \prod_{i=1}^m\left(1+\frac{1}{p_i}\right)
 \leq
 \left(\frac{4}{3}\right)^m.
\]
The resulting set \(S\) therefore satisfies
\[
 \begin{aligned}
 |S|
 &\leq
 K^2 2^{-m}D
 \prod_{i=1}^m\left(1+\frac{1}{p_i}\right)\\
 &\leq
 p_m n K^3\left(\frac{2}{3}\right)^m\\
 &\leq
 p_m n\left(\frac{2}{3}\right)^{2m/5}\\
 &\leq n,
 \end{aligned}
\]
where the last inequality holds for all sufficiently large \(m\), since \(p_m\) grows polynomially in \(m\).
Extend \(S\) arbitrarily to an \(n\)-element set \(X\subseteq[KD]\).
This is possible because \(KD\geq K^2n\geq n\).
Every progression contained in \(S\) remains in \(X\), so \(X\) has the required progression for every positive difference \(d\in[D]\).

It remains to express \(K\) and \(D\) in terms of \(n\).
The prime number theorem gives
\[
 \log D_m=\Theta(m\log m)
 \qquad\text{and}\qquad
 p_m=O(m\log m).
\]
The inequality \(n\leq D_m\) gives \(\log n=O(m\log m)\).
Conversely, \(D_{m-1}\leq nK_{m-1}\) and \(\log K_{m-1}=O(m)\) give \(\log n=\Omega(m\log m)\).
Thus \(\log n=\Theta(m\log m)\), and hence
\[
 m
 =
 \Theta\left(\frac{\log n}{\log\log n}\right).
\]
Since \(\log K=\Theta(m)\), adjusting the absolute constants gives
\[
 \exp\left(c\frac{\log n}{\log\log n}\right)
 \leq K
 \leq
 \exp\left(C\frac{\log n}{\log\log n}\right).
\]
Finally, \(p_m=O(m\log m)=(\log n)^{O(1)}\), and the choice of \(m\) gives
\[
 nK
 < D
 \leq
 p_mnK
 \leq
 nK(\log n)^C
\]
after increasing \(C\) if necessary.
Because \(D>nK\), the progressions in \(X\) include every positive integer common difference \(d\in[nK]\).
Moreover,
\[
 KD
 \leq
 nK^2(\log n)^C
 =
 n\exp\left(O\left(\frac{\log n}{\log\log n}\right)\right)
 =
 n^{1+o(1)}.
\]
Set \(u\coloneqq KD=n^{1+o(1)}\).
Then \(X\subseteq[u]\), as claimed.
\end{proof}

\end{document}